\documentclass{article}

\usepackage{amssymb,amsmath,amsthm}
\usepackage{mathrsfs}
\usepackage{epsfig}
\usepackage[usenames,dvipsnames]{color}

\newcommand\ie{{\em i.e.}}

\def\B{\mathscr B}
\def\C{\mathbb C}
\def\CC{\mathfrak C}

\def\H{\mathcal H}
\def\K{\mathcal K}

\def\N{\mathbb N}
\def\NN{\mathcal N}
\def\O{\mathcal O}

\def\R{\mathbb R}

\def\T{\mathbb T}
\def\U{\mathscr U}

\def\Z{\mathbb Z}
\def\supp{\mathrm{supp}\;\!}

\def\e{\mathop{\mathrm{e}}\nolimits}

\def\Rep{{\mathfrak{Rep}}}

\def\tomega{\boldsymbol \omega}
\def\tf{\boldsymbol f}
\def\tg{\boldsymbol g}
\def\th{\boldsymbol h}

\newtheorem{Theorem}{Theorem}[section]
\newtheorem{Remark}[Theorem]{Remark}
\newtheorem{Lemma}[Theorem]{Lemma}

\newtheorem{Proposition}[Theorem]{Proposition}
\newtheorem{Definition}[Theorem]{Definition}

\begin{document}


\title{Continuity of the spectra  \\ for families of magnetic operators on $\Z^d$}

\author{D. Parra$^1$~~and S. Richard$^2$\footnote{Supported by JSPS Grant-in-Aid for Young Scientists A
no 26707005.}}

\date{\small}
\maketitle \vspace{-1cm}

\begin{quote}
\emph{
\begin{itemize}
\item[$^1$] Universit\'e de Lyon; Universit\'e
Lyon 1; CNRS, UMR5208, Institut Camille Jordan,
43 blvd du 11 novembre 1918, F-69622
Villeurbanne-Cedex, France
\item[$^2$] Graduate school of mathematics, Nagoya University,
Chikusa-ku, Nagoya 464-8602, Japan; On leave of absence from
Universit\'e de Lyon; Universit\'e
Lyon 1; CNRS, UMR5208, Institut Camille Jordan,
43 blvd du 11 novembre 1918, F-69622
Villeurbanne-Cedex, France
\item[] \emph{E-mails:} parra@math.univ-lyon1.fr, richard@math.nagoya-u.ac.jp
\end{itemize}
}
\end{quote}


\begin{abstract}
For families of magnetic self-adjoint operators on $\Z^d$ whose symbols and magnetic fields
depend continuously on a parameter $\epsilon$, it is shown that the main spectral properties of these operators
also vary continuously with respect to $\epsilon$.
The proof is based on an algebraic setting involving twisted crossed product $C^*$-algebras.
\end{abstract}


\smallskip


\section{Introduction}
\setcounter{equation}{0}

The continuity of the spectra for families of self-adjoint operators in a Hilbert space
has been considered for several decades, but many natural questions have only received partial answers yet.
In this paper we consider a fairly general family of
magnetic Schr\"odinger operators acting on $\Z^d$ and exhibit some continuity properties of
the spectra under suitable modifications of the magnetic fields and of the symbols defining the operators.
In rough terms, the continuity we are dealing with corresponds to the stability of the spectral gaps
as well as the stability of the spectral compounds. In a more precise terminology
we shall prove inner and outer continuity for the family of spectra, as defined below.

In the discrete setting, the Harper operator is certainly the preeminent example and much efforts
have been dedicated to its study and to generalizations of this model.
It is certainly impossible to mention all papers dealing with continuity properties of families of such operators, but let us
cite a few of them which are relevant for our investigations.
First of all, let us mention the seminal paper \cite{Bel} in which the author proves the Lipschitz continuity of gap boundaries with respect to the variation
of a constant magnetic field for a family of pseudodifferential operators acting on $\Z^2$.
In \cite{Kot} and based on the framework introduced in \cite{Sun}, similar Lipschitz continuity is proved for
self-adjoint operators acting on a crystal lattice, a natural generalization of $\Z^d$.
Note that in these two references a $C^*$-algebraic framework is used, as we shall do it later on.
On the other hand, papers \cite{Nen} and \cite{Cor} deal with families of magnetic pseudodifferential operators on $\Z^2$ only but
continuity results are shown for more general symbols and magnetic fields.

Before introducing the precise framework of our investigations, let us still mention two additional papers
which are at the root of our work: \cite{MPR} in which a general framework for magnetic systems,
involving twisted crossed product $C^*$-algebras, is introduced and \cite{AMP} which contains results similar to ours but in a continuous setting.

In the Hilbert space $\H:=l^2(\Z^d)$ and for some fixed parameter $\epsilon$ let us consider operators of the form
\begin{equation}\label{eq_def_H}
[H^\epsilon u](x):=\sum_{y\in \Z^d} h^\epsilon(x;y-x)\e^{i\phi^\epsilon(x,y)}u(y)
\end{equation}
with $u\in \H$ of finite support, $x\in \Z^d$ and where $h^\epsilon: \Z^d\times \Z^d\to \C$ and $\phi^\epsilon:\Z^d\times \Z^d \to \R$ satisfy
\begin{enumerate}
\item[(i)] $\sum_{x\in \Z^d}\sup_{q\in \Z^d}|h^\epsilon(q;x)|<\infty$,
\item[(ii)] $\overline{h^\epsilon(q+x;-x)}=h^\epsilon(q;x)$ for any $q,x\in \Z^d$,
\item[(iii)] $\phi^\epsilon(x,y)=-\phi^\epsilon(y,x)$ for all $x,y\in \Z^d$.
\end{enumerate}
Such operators are usually called \emph{discrete magnetic Schr\"odinger operators}.
Note that condition (i) ensures that $H^\epsilon$ extends continuously to a bounded operator in $\H$, while conditions (ii) and (iii)
imply that the corresponding operator is self-adjoint. In the sequel a map $\phi:\Z^d\times\Z^d\to \R$ satisfying $\phi(x,y)=-\phi(y,x)$ for any $x,y\in \Z^d$
will simply be called a \emph{magnetic potential}.

Let us consider a compact Hausdorff space $\Omega$ and assume that $\epsilon \in \Omega$.
A natural question in this setting is the following: Under which regularity conditions on the maps $\epsilon\mapsto h^\epsilon$
and $\epsilon \mapsto \phi^\epsilon$ can one get some continuity for the spectra of the family of operators $\{H^\epsilon\}_{\epsilon\in \Omega}$,
and what kind of continuity can one expect on these sets ?
As already mentioned above, we shall consider the notion of inner and outer continuity, borrowed from \cite{AMP} but originally inspired by \cite{Bel}.

\begin{Definition}
Let $\Omega$ be a compact Hausdorff space, and let $\{\sigma_\epsilon\}_{\epsilon\in\Omega}$ be a family of closed subsets of $\R$.
\begin{enumerate}
\item The family $\{\sigma_\epsilon\}_{\epsilon\in \Omega}$ is \emph{outer continuous at $\epsilon_0\in \Omega$} if for any compact subset $\K$ of $\R$
such that $\K\cap \sigma_{\epsilon_0}=\emptyset$ there exists a neighbourhood $\NN=\NN(\K,\epsilon_0)$ of $\epsilon_0$ in $\Omega$ such that
$\K\cap \sigma_{\epsilon}=\emptyset$ for any $\epsilon\in \NN$,
\item The family $\{\sigma_\epsilon\}_{\epsilon\in \Omega}$ is \emph{inner continuous at $\epsilon_0\in \Omega$} if for any open subset $\O$ of $\R$
such that $\O\cap \sigma_{\epsilon_0}\neq \emptyset$ there exists a neighbourhood $\NN=\NN(\O,\epsilon_0)$ of $\epsilon_0$ in $\Omega$ such that
$\O\cap \sigma_{\epsilon}\neq\emptyset$ for any $\epsilon\in \NN$.
\end{enumerate}
\end{Definition}

Let us now present a special case of our main result which will be stated in Theorem \ref{thm_main}.
The following statement is inspired from \cite{Nen} and a comparison with the existing literature will be established just afterwards.

\begin{Theorem}\label{thm_Nenciu}
For each $\epsilon \in \Omega:=[0,1]$ let $h^\epsilon: \Z^d\times \Z^d\to \C$ satisfy the above conditions (i) and (ii).
Assume that the family $\{h^\epsilon\}_{\epsilon\in \Omega}$ satisfies for any $y\in \Z^d$ the condition
$$
\lim_{\epsilon'\to \epsilon}\sup_{q\in \Z^d}|h^{\epsilon'}(q;y)-h^\epsilon(q;y)|=0
$$
and $|h^\epsilon(q;y)|\leq f(y)$ for some $f\in l^1(\Z^d)$, all $q\in \Z^d$ and all $\epsilon \in \Omega$.
Let also $\phi$ be a magnetic potential which satisfies
$$
\big|\phi(x,y)+\phi(y,z)+\phi(z,x)\big| \leq \hbox{ area } \triangle (x,y,z),
$$
where $\triangle (x,y,z)$ means the triangle in $\R^d$ determined by the three points $x,y,z\in \Z^d$.
Then for $H^\epsilon$ defined on $u\in \H$ by
$$
[H^\epsilon u](x):=\sum_{y\in \Z^d} h^\epsilon(x;y-x)\e^{i\epsilon\phi(x,y)}u(y)
$$
the family of spectra $\sigma(H^\epsilon)$ forms an outer and an inner continuous family at every points $\epsilon\in \Omega$.
\end{Theorem}

Observe that if one considers a function $h\in l^1\big(\Z^d;l^\infty(\Z^d)\big)$ independent of $\epsilon$ and which satisfies $\overline{h(q+x;-x)}=h(q;x)$ for any $q,x\in \Z^d$,
then the various assumptions on the family $\{h^\epsilon\}_{\epsilon \in \Omega}$ are easily checked.
In \cite{Nen} the case $d=2$ is considered for a fixed symbol $h$ satisfying a decay of the form $\sup_{q\in \Z^d}|h(q;x)|\leq C\e^{-\beta|x|}$, where $0<\beta\leq 1$ and $|x|$ denotes the Euclidean norm in $\Z^2$.
In this framework, stronger continuity properties of the family of spectra are obtained, but these results deeply depend on the parameter $\beta$.
On the other hand our results are somewhat weaker but hold for a much more general class of symbols. In addition, more general $\epsilon$-dependent magnetic potentials
are considered in our main result.

Let us now emphasize that the framework presented in Section \ref{sec_field} does not allow us to get any quantitative estimate,
as emphasized in the recent paper \cite{BB}. Indeed, the very weak continuity requirement we impose on the $\epsilon$-dependence on our objets
can not lead to any Lipschitz or H\"older continuity. More stringent assumptions are necessary for that purpose, and such estimates
certainly deserve further investigations.

Our approach relies on the concepts of twisted crossed product $C^*$-algebras and on a field of such algebras, mainly borrowed from \cite{Rie,Zel}.
In the discrete setting, such algebras have already been used, for example in \cite{Bel,Kot,Sun}. However, instead of considering
a $2$-cocycle with scalar values, which is sufficient for the case of a constant magnetic field, our $2$-cocycles take values
in the group of unitary elements of $l^\infty(\Z^d)$. This allows us to consider arbitrary magnetic potential on $\Z^d$ and to encompass
all the corresponding operators in a single algebra.

Let us finally describe the content of this paper. In Section \ref{sec_magn_alg} we introduce the framework for a single magnetic system,
\ie~for a fixed $\epsilon$. For that reason, no $\epsilon$-dependence is indicated in this section. In Section \ref{sec_field} the $\epsilon$-dependence
is introduced and the continuous dependence on this parameter is studied. Our main result is presented in Theorem \ref{thm_main}.
In the last section, we provide the proof of Theorem \ref{thm_Nenciu}.

\section{Discrete magnetic systems}\label{sec_magn_alg}
\setcounter{equation}{0}

This section is divided into three parts. First of all, we motivate the introduction of the algebraic formalism by showing
that any magnetic potential leads naturally to the notion of a normalized $2$-cocycle with one additional property.
Based on this observation, we introduce in the second part of the section a special instance of a twisted crossed product $C^*$-algebra.
A faithful representation of this algebra in $l^2(\Z^d)$ is also provided.
In the third part, we draw the connections of this abstract construction with the initial magnetic system.

\subsection{From magnetic potentials to $2$-cocycles}

We start by recalling that a magnetic potential consists in a map $\phi:\Z^d\times \Z^d\to \R$ satisfying for any $x,y\in \Z^d$ the relation
\begin{equation}\label{eq_phi}
\phi(x,y)=-\phi(y,x).
\end{equation}
Then, given such a magnetic potential $\phi$ let us introduce and study a new map
$$
\omega: \Z^d\times \Z^d \times \Z^d\to \T
$$
defined for $q,x,y\in \Z^d$ by
\begin{equation}\label{def_omega}
\omega(q;x,y):=\exp\big\{i\big[\phi(q,q+x)+\phi(q+x,q+x+y)+\phi(q+x+y,q)\big]\big\}.
\end{equation}
Note that the distinction between the variable $q$ and the variables $x$ and $y$ is done on purpose.
Indeed, for fixed $x,y\in \Z^d$ we shall also use the notation $\omega(x,y)$ for the map
\begin{equation*}
\omega(x,y):\Z^d\ni q \mapsto [\omega(x,y)](q):=\omega(q;x,y) \in  \T.
\end{equation*}

Since $\Z^d$ acts on itself by translations, let us introduce the action $\theta$ of $\Z^d$ on any $f\in l^\infty(\Z^d)$ by
\begin{equation}\label{def_theta}
\theta_xf(y)=f(x+y).
\end{equation}
In particular, since $\omega(x,y)\in l^\infty(\Z^d)$ we have
\begin{equation*}
\big[\theta_z\omega(x,y)\big](q):=[\omega(x,y)](q+z)= \omega(q+z;x,y).
\end{equation*}
Based on these definitions, the following properties for $\omega$ can now be proved:

\begin{Lemma}\label{lem_2_cocycle}
Let $\phi$ be a magnetic potential and let $\omega$ defined by \eqref{def_omega}.
Then for any $x,y,z\in \Z^d$ the following properties hold:
\begin{enumerate}
\item[(i)] $\omega(x+y,z)\;\omega(x,y)  = \theta_x\omega(y,z)\;\omega(x,y+z)$,
\item[(ii)] $\omega(x,0)=\omega(0,x)=1$,
\item[(iii)] $\omega(x,-x)=1$.
\end{enumerate}
\end{Lemma}

\begin{proof}
The proof consists only in simple computations. Indeed by taking \eqref{eq_phi} into account one gets that for any $q,x,y,z\in \Z^d$
\begin{align*}
&[\omega(x+y,z)](q)\;[\omega(x,y)](q) \\
& =\omega(q;x+y,z)\;\omega(q;x,y) \\
& = \exp\big\{i\big[\phi(q,q+x+y)+\phi(q+x+y,q+x+y+z)+\phi(q+x+y+z,q)\big]\big\}  \\
& \quad \ \exp\big\{i\big[\phi(q,q+x)+\phi(q+x,q+x+y)+\phi(q+x+y,q)\big]\big\} \\
& =  \exp\big\{i\big[\phi(q+x,q+x+y)+\phi(q+x+y,q+x+y+z)+\phi(q+x+y+z,q+x)\big]\big\}  \\
& \quad \ \exp\big\{i\big[\phi(q,q+x)+\phi(q+x,q+x+y+z)+\phi(q+x+y+z,q)\big]\big\} \\
& =\omega(q+x;y,z)\;\omega(q;x,y+z) \\
& = [\theta_x\omega(y,z)](q)\;[\omega(x,y+z)](q)
\end{align*}
which proves (i). Similar computations lead to (ii) and (iii) once
the equality $\phi(x,x)=0$ for any $x\in \Z^d$ is taken into account.
\end{proof}

Let us now make some comments about the previous definitions and results.
For fixed $x,y$ the map $\omega(x,y):\Z^d\to \T$ can been seen as an element of the unitary group of the algebra $l^\infty(\Z^d)$.
For simplicity, we set $\U(\Z^d)$ for this unitary group, \ie
$$
\U(\Z^d)=\{f:\Z^d\to \T\}.
$$
In addition, property (i) of the previous lemma is usually considered as a \emph{$2$-cocycle property}
while property (ii) corresponds to a normalization of this $2$-cocycle.
In the second part of this section, we shall come back to these definitions.
For the time being, let us just mention that this $2$-cocycle will be at the root of the definition of a
twisted crossed product $C^*$-algebra.
However, before recalling the details of this construction, let us still show that $\omega$ depends only
on equivalent classes of magnetic potentials.

\begin{Lemma}\label{lem_equi_class}
Let $\phi$ be a magnetic potential and let $\varphi:\Z^d\to \R$. Then the map $\phi':\Z^d\times \Z^d\to \R$
defined by
$$
\phi'(x,y)=\phi(x,y)+\varphi(y)-\varphi(x).
$$
is a magnetic potential. In addition, by formula \eqref{def_omega} the two magnetic potentials $\phi$ and $\phi'$
define the same $2$-cocycle.
\end{Lemma}

\begin{proof}
Clearly, $\phi'(x,y)=-\phi'(y,x)$ which means that $\phi'$ is a magnetic potential.
If we denote by $\omega$ (resp. $\omega'$) the $2$-cocycle defined by \eqref{def_omega} for the magnetic potential $\phi$ (resp. $\phi'$) we get
\begin{align*}
\omega'(q;x,y)&:=\exp\big\{i\big[\phi'(q,q+x)+\phi'(q+x,q+x+y)+\phi'(q+x+y,q)\big]\big\} \\
& = \exp\big\{i\big[\phi(q,q+x)+\varphi(q+x)-\varphi(q)+\phi(q+x,q+x+y)+ \varphi(q+x+y)-\varphi(q+x) \\
& \quad \ +\phi(q+x+y,q)+\varphi(q)-\varphi(q+x+y)\big]\big\} \\
&  = \exp\big\{i\big[\phi(q,q+x)+\phi(q+x,q+x+y)+\phi(q+x+y,q)\big]\big\} \\
& = \omega(q;x,y).\qedhere
\end{align*}
\end{proof}

One could argue that the $2$-cocycle $\omega$ depends only
on the \emph{magnetic field} as introduced in \cite{CTT}, and not on the choice of a magnetic potential. However, this would lead us too far from our purpose
since we would have to consider $\Z^d$ as a graph endowed with edges between every pair of vertices.

\subsection{Twisted crossed product algebras and their representations}\label{subsec_abs}

Let us adopt a very pragmatic point of view and recall only the strictly necessary information on twisted
crossed product $C^*$-algebras. More can be found in the fundamental papers \cite{PR1,PR2} or in the review paper \cite{MPR}.
Since the group we are dealing with is simply $\Z^d$, most of the necessary information can also be found in \cite{Zel}.

Consider the group $\Z^d$ and the algebra $l^\infty(\Z^d)$ endowed with the action $\theta$ of $\Z^d$ by translations, as defined in \eqref{def_theta}.
As suggested by the notation, the vector space $l^1\big(\Z^d;l^\infty(\Z^d)\big)$ is endowed with the following norm
\begin{equation}\label{eq_norm}
\|f\|_{1,\infty}:=\sum_{x\in \Z^d}\sup_{q\in \Z^d}|f(q;x)| \qquad f\in l^1\big(\Z^d;l^\infty(\Z^d)\big),
\end{equation}
where $x$ is the variable in the $l^1$-part and $q$ is the variable in the $l^\infty$-part.
This set also admits an action of $\Z^d$ defined for any $f\in l^1\big(\Z^d;l^\infty(\Z^d)\big)$ by
$$
[\theta_y f(x)](q):=[f(\cdot +y;x)](q) = f(q+y;x).
$$

In order to endow $l^1\big(\Z^d;l^\infty(\Z^d)\big)$ with a twisted product,
let $\omega$ be any normalized $2$-cocycle on $\Z^d$ with values in the unitary group of $l^\infty(\Z^d)$,
or in other words let $\omega:\Z^d\times \Z^d\to \U(\Z^d)$ satisfy for any $x,y,z\in \Z^d$:
\begin{equation}\label{eq_2}
\omega(x+y,z)\;\omega(x,y) = \theta_x\omega(y,z)\;\omega(x,y+z)
\end{equation}
and
\begin{equation}\label{eq_n}
\omega(x,0)=\omega(0,x)=1.
\end{equation}
Because of the point (iii) of Lemma \ref{lem_2_cocycle}, we shall also assume that the $2$-cocycle $\omega$
satisfies an additional property, namely for any $x\in \Z^d$:
\begin{equation}\label{eq_add}
\omega(x,-x)=1.
\end{equation}

We can now define the twisted product and an involution: for any $f,g\in l^1\big(\Z^d;l^\infty(\Z^d)\big)$
one sets
\begin{equation}\label{eq_produit}
[f\diamond g](x):=\sum_{y\in \Z^d} f(y)\; \theta_y g(x-y) \;\omega(y,x-y)
\end{equation}
and
\begin{equation}\label{eq_involution}
f^{\diamond}(x)= [\theta_x f(-x)]^*=\overline{f(\cdot+x;-x)}.
\end{equation}
Both operations are continuous with respect to the norm introduced in \eqref{eq_norm}.

The enveloping $C^*$-algebra of $l^1\big(\Z^d;l^\infty(\Z^d)\big)$, endowed with the above product and involution,
will be denoted by $\CC(\omega)\equiv\CC$. Recall that this algebra corresponds to the completion of $l^1\big(\Z^d;l^\infty(\Z^d)\big)$
with respect to the $C^*$-norm defined as the supremum over all the faithful representations of $l^1\big(\Z^d;l^\infty(\Z^d)\big)$.
As a consequence, $l^1\big(\Z^d;l^\infty(\Z^d)\big)$
is dense in $\CC$ and the new $C^*$-norm $\|\cdot\|$ satisfies $\|f\|\leq \|f\|_{1,\infty}$.

\begin{Remark}
In \cite{MPR} an additional ingredient is introduced in the previous construction, namely an endomorphism $\tau$ of $\Z^d$.
In the continuous case, when $\Z^d$ is replaced by $\R^d$, this additional degree of freedom allows one
to encompass in a single framework the formulas for the Weyl quantization and for the Kohn-Nirenberg quantization.
In the discrete setting, we stick to the case $\tau=0$ since the other choices
do not seem to be relevant.
\end{Remark}

Let us now look at a faithful representation of the algebra $\CC$ in the Hilbert space $\H=l^2(\Z^d)$.
First of all, by \cite[Lem.~2.9]{MPR} there always exists a \emph{$1$-cochain} $\lambda$, \ie~a map $\lambda: \Z^d\to \U(\Z^d)$, such that
\begin{equation}\label{eq_1_cochain}
\lambda(x)\; \theta_x\lambda(y)\; \lambda(x+y)^{-1} = \omega(x,y).
\end{equation}
In fact, an example of such a $1$-cochain can be defined by the following formula:
\begin{equation}\label{eq_lambda_t}
\lambda_t(q;x)\equiv[\lambda_t(x)](q):=\omega(0;q,x).
\end{equation}
Indeed, it easily follows from the $2$-cocycle property \eqref{eq_2} that
\begin{align*}
\lambda_t(q;x)\;\!\lambda_t(q+x;y)\;\!\lambda_t(q;x+y)^{-1}
& = \omega(0;q,x)\;\!\omega(0;q+x,y)\;\! \omega(0;q,x+y)^{-1} \\
&= \theta_q\omega(0;x,y)\\
& = \omega(q;x,y).
\end{align*}
Note that in the continuous case this choice corresponds to the transversal gauge for the magnetic potential, and this is why the index $t$ has been added.

Since the $2$-cocycle $\omega$ has been chosen normalized and with the additional property \eqref{eq_add}, the
$1$-cochains satisfying \eqref{eq_1_cochain} also share some additional properties, namely:

\begin{Lemma}\label{lem_sur_lambda}
Let $\lambda$ be a $1$-cochain satisfying \eqref{eq_1_cochain} for $\omega$ satisfying \eqref{eq_2}-\eqref{eq_add}. Then,
\begin{enumerate}
\item[(i)] $\lambda(q;0)=1$ for any $q\in \Z^d$,
\item[(ii)] $\lambda (y;x-y)=\lambda(x;y-x)^{-1}$ for any $x,y\in \Z^d$.
\end{enumerate}
\end{Lemma}

\begin{proof}
One infers from \eqref{eq_1_cochain} for $y=0$ and from \eqref{eq_n} that
$$
\lambda(q;x)\; \lambda(q+x;0)\; \lambda(q;x)^{-1} = \lambda(q+x;0) = \omega(q;x,0)=1.
$$
Similarly, from \eqref{eq_1_cochain} for $y=-x$ and from \eqref{eq_add} one gets that
\begin{equation*}
\lambda(x)\;\!\theta_x\lambda(-x)\;\!\lambda(0)=\omega(x,-x)=1,
\end{equation*}
from which one deduces that $\lambda(q;x) = \lambda(q+x;-x)^{-1}$. Finally, by replacing $q$ by $y$ and $x$ by $x-y$ in the previous equality
one deduces the statement.
\end{proof}

Once a $1$-cochain satisfying \eqref{eq_1_cochain} has been chosen, a representation of $\CC$ in $\H$ can be defined, as shown in  \cite[Sec.~2.4]{MPR}.
More precisely, for any $h\in l^1\big(\Z^d;l^\infty(\Z^d)\big)$, any $u\in \H$ and any $x\in \Z^d$ one sets
\begin{equation*}
[\Rep^\lambda(h)u](x) := \sum_{y\in \Z^d} h(x;y-x) \;\!\lambda(x;y-x) \;\!u(y).
\end{equation*}
The main properties of this representation are gathered in the following statement, which corresponds to \cite[Prop.~2.16 \& 2.17]{MPR}
adapted to our setting. In (i) the operator $\varphi(X)$ denotes the operator of multiplication by the function $\varphi$.

\begin{Proposition}\label{prop_MPR}
Let $\lambda$ and $\lambda'$ be two $1$-cochains satisfying \eqref{eq_1_cochain} for the same $\omega$ that satisfies \eqref{eq_2}-\eqref{eq_add}. Then,
\begin{enumerate}
\item[(i)] There exists $\varphi:\Z^d\to \R$ such that
$$\lambda'(q;x)=\e^{i\theta_x \varphi(q)}\e^{-i\varphi(q)}\lambda(q;x) .$$
In addition one has for any  $h\in l^1\big(\Z^d;l^\infty(\Z^d)\big)$
$$
\Rep^{\lambda'}(h) = \e^{-i\varphi(X)}\;\! \Rep^\lambda(h)\;\!\e^{i\varphi(X)},
$$
\item[(ii)] The representation $\Rep^\lambda$ is irreducible,
\item[(iii)] The representation $\Rep^\lambda$ is faithful.
\end{enumerate}
\end{Proposition}

Let us end this abstract part with a result about self-adjointness. The following statement shows that
if $h\in l^1\big(\Z^d;l^\infty(\Z^d)\big)$ satisfies $h^\diamond=h$, with the involution defined in \eqref{eq_involution},
then the corresponding operator $\Rep^\lambda(h)$ is self-adjoint.

\begin{Lemma}
Let $\lambda$ be any $1$-cochain satisfying \eqref{eq_1_cochain} with $\omega$ satisfying \eqref{eq_2}-\eqref{eq_add},
and let $h\in l^1\big(\Z^d;l^\infty(\Z^d)\big)$. Then $\Rep^\lambda(h)$ is self-adjoint if $h^\diamond=h$.
\end{Lemma}

\begin{proof}
Let $u,v$ be elements of the Hilbert space $l^2(\Z^d)$ with compact support, let $\langle \cdot,\cdot\rangle$ denote its scalar product
and let $\langle\cdot,\cdot\rangle_\C$ denote the scalar product in $\C$.
Let us also observe that with a simple change of variables the equality $h^\diamond = h$ is equivalent to $h(y;x-y)=\overline{h(x;y-x)}$.
Then by taking Lemma \ref{lem_sur_lambda}.(ii) into account one gets
\begin{align*}
\big\langle v, \Rep^\lambda(h)u\big\rangle
= & \sum_{x\in \Z^d} \Big\langle v(x), \sum_{y\in \Z^d} h(x;y-x) \;\!\lambda(x;y-x) \;\!u(y)\Big\rangle_\C \\
= &  \sum_{y\in \Z^d}  \Big\langle  \sum_{x\in \Z^d} \overline{h(x;y-x)} \;\!\overline{\lambda(x;y-x)}\;\!v(x), u(y)\Big\rangle_\C \\
= &  \sum_{x\in \Z^d}  \Big\langle  \sum_{y\in \Z^d} h(x;y-x) \;\!\lambda(x;y-x)\;\!v(y), u(x)\Big\rangle_\C \\
= & \big\langle \Rep^\lambda(h)v,u\big\rangle.\qedhere
\end{align*}
\end{proof}

\subsection{Back to magnetic systems}

Let us now come back to a magnetic potential $\phi$ and to the \emph{magnetic $2$-cocycle} $\omega$ defined by \eqref{def_omega}.
By Lemma \ref{lem_2_cocycle}, the three conditions \eqref{eq_2}-\eqref{eq_add} are satisfied for such a $2$-cocycle, and thus
the construction of Section \ref{subsec_abs} is at hand. Let us thus list some relations between this abstract section and
some magnetic objects considered before.

First of all, the relation between $\lambda_t$ introduced in \eqref{eq_lambda_t} and $\phi$ can be explicitly computed, namely
\begin{align}\label{eq_lambda_phi}
\nonumber \lambda_t(q;x) & = \omega(0;q,x) \\
\nonumber & = \exp\big\{i\big[\phi(0,q)+\phi(q,q+x)+\phi(q+x,0)\big]\big\} \\
\nonumber & = \exp\big\{i\big[\phi(q,q+x)+ \phi(q+x,0)-\phi(q,0)\big]\big\} \\
& = \exp\big\{i\big[\phi(q,q+x)+ \varphi(q+x)-\varphi(q)\big]\big\}
\end{align}
with $\varphi:\Z^d\to \R$ defined by $\varphi(x):=\phi(x,0)$.
On the other hand, the obvious choice
\begin{equation}\label{eq_obvious}
\lambda_\phi(q;x):=\e^{i\phi(q,q+x)}
\end{equation}
is also a $1$-cochain satisfying \eqref{eq_1_cochain}, as a consequence of \eqref{def_omega} and $\phi(x,y)=-\phi(y,x)$.

At the level of the representations, for the $1$-cochain $\lambda_\phi$ one gets
\begin{equation}\label{eq_Rep_phi}
[\Rep^{\lambda_\phi}(h)u](x) = \sum_{y\in \Z^d} h(x;y-x) \;\!\e^{i\phi(x,y)} \;\!u(y).
\end{equation}
Clearly, this expression corresponds to the one provided in \eqref{eq_def_H} which
was the starting point of our investigations.
It is precisely the equality of these two expressions which makes the algebraic formalism useful
for the study of magnetic operators.

On the other hand for the $1$-cochain $\lambda_t$ and if \eqref{eq_lambda_phi} is taken into account one obtains
\begin{align*}
[\Rep^{\lambda_t}(h)u](x) & = \sum_{y\in \Z^d} h(x;y-x) \;\!\exp\{i\phi(x,y) + \varphi(y)-\varphi(x)\} \;\!u(y) \\
& = \e^{-i\varphi(x)} \sum_{y\in \Z^d} h(x;y-x) \;\!\e^{i\phi(x,y)}\;\!\e^{i\varphi(y)} \;\!u(y) \\
& = \big[\e^{-i\varphi(X)} \Rep^{\lambda_\phi}(h) \e^{i\varphi(X)}u\big](x).
\end{align*}
These equalities mean that the representations provided by $\Rep^{\lambda_\phi}$ and $\Rep^{\lambda_t}$ are unitarily equivalent,
as it could already be inferred from Proposition \ref{prop_MPR}.(i).

In summary, any magnetic potential defines a magnetic $2$-cocycle, and subsequently a twisted crossed product $C^*$-algebra
which can be represented faithfully in $\H$. This algebra depends on an equivalence class of magnetic potentials, as emphasized
in Lemma \ref{lem_equi_class}. Reciprocally, any normalized $2$-cocycle on $\Z^d$ with values in $\U(\Z^d)$ and which satisfies
the additional relation \eqref{eq_add} comes from a magnetic potential, as shown in the following lemma.

\begin{Lemma}
Let $\omega:\Z^d\times\Z^d\to \U(\Z^d)$ satisfy conditions \eqref{eq_2}-\eqref{eq_add}.
Then there exists a magnetic potential which satisfies the relation \eqref{def_omega}.
\end{Lemma}

\begin{proof}
First of all, observe that the equality
\begin{equation}\label{eq_rel_1}
\omega(x,y)=\omega(x+y,-y)^{-1}
\end{equation}
is a direct consequence of \eqref{eq_2}-\eqref{eq_add}
taking $z=-y$ in \eqref{eq_2}.

For any $x,y\in \Z^d$ with $\omega(0;x,y-x)\neq -1$ let us set $\phi(x,y)\in (-\pi,\pi)$ by
$$
\e^{i \phi(x,y)}:=\omega(0;x,y-x).
$$
By \eqref{eq_rel_1} one infers that
$$
\e^{i\phi(y,x)}=\omega(0;y,x-y) = \omega(0;x,y-x)^{-1} = \big(\e^{i\phi(x,y)}\big)^{-1} = \e^{-i\phi(x,y)}
$$
which means that $\phi(x,y)=-\phi(y,x)$. If $\omega(0;x,y-x)= -1$, then one sets $\phi(x,y):=-\pi$ if $x<y$ (lexicographic order on $\Z^d$)
while $\phi(x,y):=\pi$ if $y<x$. With this convention and by the same argument one obtains that  $\phi(x,y)=-\phi(y,x)$ which is thus
proved for any $x,y\in \Z^d$. As a consequence, $\phi$ is indeed a magnetic potential.

In order to show \eqref{def_omega}, it is enough to observe that
\begin{align*}
& \exp\big\{i\big[\phi(q,q+x)+\phi(q+x,q+x+y)+\phi(q+x+y,q)\big]\big\} \\
& = \omega(0;q,x)\;\!\omega(0;q+x,y)\;\!\omega(0;q+x+y,-x-y) \\
& = \omega(0;q,x)\;\!\omega(0;q+x,y)\;\!\omega(0;q,x+y)^{-1} \\
& = \omega(q;x,y)
\end{align*}
where \eqref{eq_rel_1} has again been used for the second equality, and where the $2$-cocycle property \eqref{eq_2}
has been taken into account for the last equality.
\end{proof}

\section{A continuous field of $C^*$-algebras}\label{sec_field}
\setcounter{equation}{0}

In this section we consider a family of discrete magnetic systems which are parameterized by the elements $\epsilon$
of a compact Hausdorff space $\Omega$. The necessary continuity relations between
the various objects is encoded in the structure of a field of twisted crossed product $C^*$-algebras,
as introduced in \cite{Rie} and already used in a similar context in \cite{AMP}.

Our first aim is to recall the notion of a continuous field of $2$-cocycles \cite[Def.~2.1]{Rie}.
In our framework, with the locally compact group $\Z^d$ and the algebra $l^\infty(\Z^d)$, we get the following definition.

\begin{Definition}\label{def_cont_omega}
A \emph{continuous field over $\Omega$ of $2$-cocycles on $\Z^d$} is a map
$$
\tomega: \Omega\times \Z^d\times \Z^d \to \U(\Z^d)
$$
such that
\begin{enumerate}
\item[(i)] For any $x,y,z\in \Z^d$ and $\epsilon \in \Omega$ the following relations hold:
\begin{equation}\label{eq_cont_field}
\tomega(\epsilon;x+y,z)\;\tomega(\epsilon;x,y)  = \theta_x\tomega(\epsilon;y,z)\;\tomega(\epsilon;x,y+z)
\end{equation}
and
\begin{equation}\label{eq_norm_field}
\tomega(\epsilon;x,0) =\tomega(\epsilon;0,x)=1,
\end{equation}
\item[(ii)] For any fixed $(x,y)\in \Z^d\times \Z^d$ the map
\begin{equation}\label{cont_cocycle}
\Omega\ni \epsilon \mapsto \tomega(\epsilon;x,y)\in \U(\Z^d)
\end{equation}
is continuous.
\end{enumerate}
\end{Definition}

Note that in equation \eqref{eq_cont_field} the shift $\theta_x$ acts on
$\tomega(\epsilon;x,y)\in\U(\Z^d)$, as in the previous section.
Clearly, for each fixed $\epsilon$ the relation \eqref{eq_cont_field} corresponds to \eqref{eq_2} in the abstract framework or to the statement (i) of Lemma \ref{lem_2_cocycle}
in the magnetic case.
Similarly, \eqref{eq_norm_field} is a reminiscence of \eqref{eq_n} or of the normalization property (ii) of Lemma \ref{lem_2_cocycle}.
On the other hand, the new assumption \eqref{cont_cocycle} is the one which provides the necessary continuity condition.

For the subsequent algebraic construction we shall consider the algebra $C\big(\Omega;l^\infty(\Z^d)\big)$
which is going to replace the algebra $l^\infty(\Z^d)$ of the previous section.
Note that for any $f\in C\big(\Omega;l^\infty(\Z^d)\big)$ the action of $\Z^d$ by translations is defined by $\theta_x f(q,\epsilon)=f(q+x,\epsilon)$
for any $q,x\in \Z^d$ and $\epsilon \in \Omega$.
So, let us consider a continuous field $\tomega$ over $\Omega$ of $2$-cocycles on $\Z^d$ and observe that its definition is made such that it can be
interpreted as a normalized $2$-cocycle on $\Z^d$ taking values in the unitary group of the Abelian algebra $C\big(\Omega;l^\infty(\Z^d)\big)$. Indeed,
for fixed $x,y\in \Z^d$ one can set
\begin{equation*}
[\omega(x,y)](q,\epsilon):=[\tomega(\epsilon;x,y)](q)\equiv\tomega(q,\epsilon;x,y).
\end{equation*}
Then, by condition (ii) of Definition \ref{def_cont_omega} one infers that
$$
\omega(x,y)\in C\big(\Omega;l^\infty(\Z^d)\big)\qquad \hbox{ and } \qquad [\omega(x,y)](q,\epsilon)\in \T.
$$
In addition, \eqref{eq_norm_field} implies that $\omega$ is a normalized $2$-cocycle.
Note that the additional property
\begin{equation}\label{eq_cond_add}
\omega(x,-x)=1
\end{equation}
holds if and only if $\tomega(\epsilon;x,-x)=1$ for every $\epsilon \in \Omega$.
Since this property is satisfied by magnetic $2$-cocycles we shall assume it in the sequel.

In summary, starting from a continuous field $\tomega$ over $\Omega$ of $2$-cocycles on $\Z^d$ which satisfies the additional property
$\tomega(\epsilon;x,-x)=1$ for any $\epsilon\in \Omega$ and $x\in \Z^d$, we end up with the $2$-cocycle $\omega$
taking values in the unitary group of the algebra $C\big(\Omega;l^\infty(\Z^d)\big)$
and having the additional property \eqref{eq_cond_add}.
With this $2$-cocycle one defines in analogy with \eqref{eq_produit} the product for any
$\tf,\tg\in l^1\big(\Z^d;C\big(\Omega;l^\infty(\Z^d)\big)\big)$ by
\begin{align*}
\big[[\tf\diamond \tg](x)\big](q,\epsilon) & \equiv [\tf\diamond \tg](q,\epsilon;x)\\
\nonumber & := \sum_{y\in \Z^d} [\tf(y)](q,\epsilon)\; [\tg(x-y)](q+y,\epsilon) \;[\omega(y,x-y)](q,\epsilon) \\
\nonumber & \equiv \sum_{y\in \Z^d} \tf(q,\epsilon;y)\; \tg(q+y,\epsilon;x-y) \;\tomega(q,\epsilon;y,x-y), \quad \forall \;\! q,x\in \Z^d, \epsilon\in \Omega.
\end{align*}
We also endow $l^1\big(\Z^d;C\big(\Omega;l^\infty(\Z^d)\big)\big)$ with the involution
\begin{equation*}
[\tf^\diamond(x)](q,\epsilon)\equiv \tf^{\diamond}(q,\epsilon;x):= \overline{\tf(q+x,\epsilon;-x)}
\end{equation*}
and with the norm
$$
\|\tf\|_{1,\infty}:=\sum_{x\in \Z^d}\sup_{q\in \Z^d}\sup_{\epsilon\in \Omega}|\tf(q,\epsilon;x)| \qquad \tf\in l^1\big(\Z^d;C\big(\Omega;l^\infty(\Z^d)\big)\big),
$$
making it a unital Banach $^*$-algebra. The enveloping $C^*$-algebra of $l^1\big(\Z^d;C\big(\Omega;l^\infty(\Z^d)\big)\big)$ will be denoted by $\CC_\Omega$.

Let us now emphasize the main point of all this construction: there exists an \emph{evaluation map}
$$
e_\epsilon: l^1\big(\Z^d;C\big(\Omega;l^\infty(\Z^d)\big)\big) \to l^1\big(\Z^d;l^\infty(\Z^d)\big)
$$
defined for any $\tf\in l^1\big(\Z^d;C\big(\Omega;l^\infty(\Z^d)\big)\big)$ by $[e_\epsilon (\tf)](q;x);=\tf(q,\epsilon;x)$
for any $q,x\in \Z^d$ and $\epsilon\in\Omega$.
This map is clearly norm-decreasing and surjective, and extends continuously to a norm-decreasing $*$-homomorphism
$e_\epsilon:\CC_\Omega\to \CC_\epsilon$, with $\CC_\epsilon:=\CC(\omega_\epsilon)$ the $C^*$-algebra constructed in the previous section
with the $2$-cocycle $\omega_\epsilon:=\tomega(\epsilon;\cdot,\cdot):\Z^d\times\Z^d\to \U(\Z^d)$.

In this framework, the main result borrowed from \cite{Rie} reads:

\begin{Proposition}\label{prop_Rieffel}
Let $\tomega$ be a continuous field over $\Omega$ of $2$-cocycles on $\Z^d$ satisfying $\tomega(\epsilon;x,-x)=1$ for
every $\epsilon \in \Omega$ and $x\in \Z^d$.
Then the following properties hold:
\begin{enumerate}
\item[(i)] The map $e_\epsilon: \CC_\Omega\to \CC_\epsilon$ is surjective,
\item[(ii)] For any $\tf\in \CC_\Omega$ one has $\|\tf\|_{\CC_\Omega} = \sup_{\epsilon\in\Omega}\|e_\epsilon(\tf)\|_{\CC_\epsilon}$,
\item[(iii)] For any $\tf\in \CC_\Omega$ the map $\Omega\ni \epsilon\mapsto \|e_\epsilon(\tf)\|_{\CC_\epsilon} \in \R_+$ is continuous.
\end{enumerate}
\end{Proposition}

\begin{proof}
Before mentioning the precise arguments borrowed from \cite{Rie}, let us stress that part of the proofs in that reference
relies on the existence of an bounded approximate identity. However, this technical point is automatically satisfied in our framework,
as shown in the seminal paper \cite[Sec.~2.28 \& 2.29]{Zel}.

Once this preliminary observation is taken into account, statement (i) is a direct consequence of \cite[Prop.~2.3]{Rie}.
For (ii) it is enough to observe that the map $\tf\mapsto \oplus_\epsilon e_\epsilon(\tf)$  is injective.
The upper semi-continuity of the map mentioned in (iii) follows from \cite[Thm.~2.4]{Rie} while the lower semi-continuity
of the same map is a consequence of \cite[Thm.~2.5]{Rie} together with the equality between the twisted crossed product algebra $\CC_\Omega$
and its reduced version, see \cite[Thm.~5.1]{Zel}.
\end{proof}

Let us now state and prove our main result:

\begin{Theorem}\label{thm_main}
Let $\tomega$ be a continuous field over $\Omega$ of $2$-cocycles on $\Z^d$ satisfying $\tomega(\epsilon;x,-x)=1$ for any $\epsilon\in \Omega$ and $x\in \Z^d$,
and let $\omega_\epsilon$ be defined by $\tomega(\epsilon;\cdot,\cdot)$.
Consider a family $\{h^\epsilon\}_{\epsilon\in \Omega}\subset l^1\big(\Z^d;l^\infty(\Z^d)\big)$ such that the following conditions are satisfied:
\begin{enumerate}
\item[(i)] For any $y\in \Z^d$, $\sup_{q\in \Z^d} \big|h^\epsilon(q;y)-h^{\epsilon'}(q;y)\big|\to 0$ as $\epsilon'\to \epsilon$ in $\Omega$,
\item[(ii)] $\sum_{y\in \Z^d} \sup_{\epsilon\in \Omega}\sup_{q\in \Z^d}|h^\epsilon(q;y)|<\infty$.
\item[(iii)]  $(h^\epsilon)^\diamond=h^\epsilon$.
\end{enumerate}
Then, for any family of $1$-cochains $\lambda^\epsilon: \Z^d\to \U(\Z^d)$ satisfying
$\lambda^\epsilon(x)\; \theta_x\lambda^\epsilon(y)\; \lambda^\epsilon(x+y)^{-1} = \omega_\epsilon(x,y)$,
the family of spectra $\Big\{\sigma\big(\Rep^{\lambda^\epsilon}(h^\epsilon)\big)\Big\}_{\epsilon \in \Omega}$ forms an outer and an inner continuous family
at every point of $\Omega$.
\end{Theorem}

Let us stress that condition (ii) is satisfied for example if there exists $f\in l^1(\Z^d)$ such that $|h^\epsilon(q;y)|\leq f(y)$
for any $\epsilon \in \Omega$ and $q\in \Z^d$.
In the following proof, we use the notation $C_0(\R)$ for the set of continuous functions on $\R$ which vanish at infinity.
Recall also that $\H=l^2(\Z^d)$.

\begin{proof}
a) Let us first observe that the conditions on $\{h^\epsilon\}_{\epsilon\in \Omega}$ have been chosen such that the function $\th:\Omega\times\Z^d\times \Z^d\to \C$
defined by $\th(\epsilon,q;x):=h^\epsilon(q;x)$ satisfies $\th\in l^1\big(\Z^d;C\big(\Omega;l^\infty(\Z^d)\big)\big)$.
As a consequence, $\th\in \CC_\Omega$ and the statements of Proposition \ref{prop_Rieffel} hold for $\th$ instead of $\tf$.
In particular, since $e_\epsilon(\th)=h^\epsilon$ one infers that the map
$$
\Omega\ni \epsilon\mapsto \|h^\epsilon\|_{\CC_\epsilon} \in \R_+
$$
is continuous. Furthermore, since $\th^{\diamond}=\th$ and since the algebra $l^1\big(\Z^d;C\big(\Omega;l^\infty(\Z^d)\big)\big)$ is unital,
one infers that for $z\in \C\setminus\R$ the element $\th-z$ is invertible in the $C^*$-algebra $\CC_\Omega$, with respect to the product $\diamond$.
Its inverse is simply denoted by $(\th-z)^{-1}$.
Since $e_\epsilon$ is a $*$-homomorphism, one gets that $e_\epsilon\big((\th-z)^{-1}\big)=(h^\epsilon-z)^{-1}$, and as a consequence the map
$$
\Omega \ni \epsilon \mapsto \big\|(h^\epsilon-z)^{-1}\big\|_{\CC_\epsilon}\in \R_+
$$
is also continuous.

Now, since $\Rep^{\lambda^\epsilon}$ defines a faithful representation of $\CC_\epsilon$ in $\B(\H)$ it follows that
$$
\|h^\epsilon\|_{\CC_\epsilon} = \big\|\Rep^{\lambda^\epsilon}(h^\epsilon)\big\|
$$
where the norm on the r.h.s.~corresponds to the usual norm in $\B(\H)$. If we set $H^\epsilon:= \Rep^{\lambda^\epsilon}(h^\epsilon)$,
which is a self-adjoint element of the $C^*$-algebra
$\Rep^{\lambda^\epsilon}(\CC_\epsilon)\subset \B(\H)$, one then deduces that the map
$$
\Omega\ni \epsilon\mapsto \|(H^\epsilon-z)^{-1}\|\in \R_+
$$
is continuous for any $z\in \C\setminus \R$.
Finally, by a density argument of the linear span of $\{(\cdot-z)^{-1}\}_{z\in \C\setminus \R}$, one
infers that for any $\eta\in C_0(\R)$, the map
\begin{equation}\label{eq_cont_eta}
\Omega\ni \epsilon \mapsto \|\eta(H^\epsilon)\|\in \R_+
\end{equation}
is continuous, see \cite[p.~364]{ABG}.

b) It remains to show that the continuity \eqref{eq_cont_eta} implies the inner and the outer continuity for the family of spectra
$\sigma(H^\epsilon)$. The following argument is directly borrowed from the proof of \cite[Prop.~2.5]{AMP} which we present here for the sake of completeness.
For the outer continuity, let $\epsilon_0\in \Omega$ and let $\K$ be a compact set in $\R$ such that $\K\cap \sigma(H^{\epsilon_0})=\emptyset$.
By Urysohn's lemma there exists $\eta\in C_0(\R)$ with $\eta\geq 0$ such that $\eta|_\K = 1$ and $\eta|_{\sigma(H^{\epsilon_0})}=0$, and therefore
$\eta(H^{\epsilon_0})=0$. By the continuity of \eqref{eq_cont_eta} one can then chose a neighbourhood $\NN$ of $\epsilon_0$ in $\Omega$ such that
for any $\epsilon\in \NN$ one has $\|\eta(H^\epsilon)\|\leq 1/2$.
By contradiction, if for some $\epsilon \in \NN$ one would have $\nu\in \K\cap \sigma(H^\epsilon)$ then  it would follow that
$$
1=\eta(\nu)\leq \sup_{\mu\in \sigma(H^\epsilon)}\eta(\mu) \leq \|\eta(H^\epsilon)\|\leq 1/2
$$
which is absurd. Thus, the family $\{\sigma(H^\epsilon)\}_{\epsilon\in \Omega}$ is outer continuous at every points of $\Omega$.

For the inner continuity, let $\O$ be an open subset of $\R$ such that there exists $\nu\in \O\cap \sigma(H^{\epsilon_0})$.
By Urysohn's lemma there exists $\eta\in C_0(\R)$ with $\eta(\nu)=1$ and $\supp \eta\subset \O$. As a consequence, one has
$\|\eta(H^{\epsilon_0})\|\geq 1$. By contradiction, assume now that for any neighbourhood $\NN$ of $\epsilon_0$ in $\Omega$ there exists $\epsilon \in \NN$ such that
$\O\cap \sigma(H^\epsilon)=\emptyset$. It follows that $\eta(H^\epsilon)=0$. However, this clearly contradicts the continuity provided in \eqref{eq_cont_eta}.
As a consequence, the family $\{\sigma(H^\epsilon)\}_{\epsilon\in \Omega}$ is inner continuous at every points of~$\Omega$.
\end{proof}

\section{The scaling example}\label{sec_Nenciu}
\setcounter{equation}{0}

In this section we provide the proof of Theorem \ref{thm_Nenciu}. On the way we also show that the setting considered in \cite{Nen}
is covered by our formalism. Note that only $d=2$ is considered in that reference, but that the extension for
arbitrary $d\in \N$ is harmless in our framework.

Let us now fix $\Omega=[0,1]$. In the following statement, $\triangle (x,y,z)$ denotes the triangle in $\R^d$ defined
by the three points $x,y,z\in \Z^d$.

\begin{Proposition}\label{prop_triangle}
Let $\phi$ be a magnetic potential which satisfies
$$
\big|\phi(x,y)+\phi(y,z)+\phi(z,x)\big| \leq \hbox{ area } \triangle (x,y,z).
$$
Then the map $\tomega$ defined for $\epsilon \in [0,1]$ and $q,x,y\in \Z^d$ by
\begin{equation}\label{eq_def_tomega}
\tomega(q,\epsilon;x,y) :=\exp\big\{i\epsilon \big[\phi(q,q+x)+\phi(q+x,q+x+y)+\phi(q+x+y,q)\big]\big\}\in \T
\end{equation}
is a continuous field over $[0,1]$ of $2$-cocycles on $\Z^d$.
\end{Proposition}

\begin{proof}
Clearly, for each fixed $\epsilon \in [0,1]$ the map $\omega(\epsilon;\cdot,\cdot)$ satisfied the normalized $2$-cocycle requirements as
mentioned in the point (i) of Definition \ref{def_cont_omega}.
For the condition (ii) of this definition, let $x,y\in \Z^d$ be fixed, and let $\epsilon,\epsilon'\in [0,1]$ and $q\in \Z^d$.
Then one has
\begin{align*}
& \big|\tomega(q,\epsilon';x,y)-\tomega(q,\epsilon;x,y)\big| \\
& \leq \Big|1-\exp\big\{i(\epsilon-\epsilon') \big[\phi(q,q+x)+\phi(q+x,q+x+y)+\phi(q+x+y,q)\big]\big\}\Big| \\
& \leq \sum_{n=1}^\infty \frac{|\epsilon-\epsilon'|^n}{n!}\;\! \Big|\phi(q,q+x)+\phi(q+x,q+x+y)+\phi(q+x+y,q)\Big|^n \\
& = \sum_{n=1}^\infty \frac{|\epsilon-\epsilon'|^n}{n!}\;\!\big|\hbox{ area } \triangle (q,q+x,q+x+y)\big|^n \\
& = \sum_{n=1}^\infty \frac{|\epsilon-\epsilon'|^n}{n!}\;\!\big|\hbox{ area } \triangle (0,x,x+y)\big|^n.
\end{align*}
Since the last expression is independent of $q$ one directly infers that
$$
\lim_{\epsilon'\to\epsilon} \big\|\tomega(\epsilon';x,y)-\tomega(\epsilon;x,y)\big\|_{\U(\Z^d)}=
\lim_{\epsilon'\to\epsilon} \sup_{q\in \Z^d}\big|\tomega(q,\epsilon';x,y)-\tomega(q,\epsilon;x,y)\big|=0
$$
which corresponds to the required continuity property.
\end{proof}

We can now show that Theorem \ref{thm_Nenciu} is a special case of Theorem \ref{thm_main}.

\begin{proof}[Proof of Theorem \ref{thm_Nenciu}]
Let us first observe that the initial conditions (i) and (ii) imposed on $h^\epsilon$ are equivalent to $h^\epsilon \in l^1\big(\Z^d;l^\infty(\Z^d)\big)$ and to
$(h^\epsilon)^\diamond=h^\epsilon$.
These two conditions together with the other two conditions imposed on $h^\epsilon$
imply that all assumptions on $h^\epsilon$ required by Theorem \ref{thm_main} are satisfied.

On the other hand, the condition imposed on the magnetic potential $\phi$ are precisely the one already used
in Proposition \ref{prop_triangle} for defining the continuous field $\tomega$ of $2$-cocycles in \eqref{eq_def_tomega}.
Observe also that the additional condition $\tomega(\epsilon;x,-x)=1$ holds for any $x\in \Z^d$.
Thus, we have checked so far that all assumptions of Theorem \ref{thm_main} are satisfied. It only remains to exploit its consequences.

For that purpose let us choose the $1$-cochain $\lambda^\epsilon$ defined as in \eqref{eq_obvious} by
$$
\lambda^\epsilon(q;x):=\e^{i\epsilon\phi(q,q+x)}.
$$
With this choice and as in \eqref{eq_Rep_phi} one obtains on any $u\in \H$
\begin{equation*}
[\Rep^{\lambda^\epsilon}(h^\epsilon)u](x) = \sum_{y\in \Z^d} h^\epsilon(x;y-x) \;\!\e^{i\epsilon\phi(x,y)} \;\!u(y) \equiv [H^\epsilon u](x)
\end{equation*}
or equivalently $\Rep^{\lambda^\epsilon}(h^\epsilon)=H^\epsilon$.
The statement of Theorem \ref{thm_Nenciu} can now be directly deduced from Theorem \ref{thm_main}.
\end{proof}



\end{document}